\documentclass[pra,twocolumn,preprintnumbers,amsmath,amssymb,superscriptaddress]{revtex4}
\usepackage[]{graphicx}
\usepackage[]{amsmath}
\usepackage[]{amsthm}
\usepackage{hyperref}
\usepackage{verbatim}
\DeclareMathOperator{\Tr}{Tr}
\DeclareMathOperator{\bRp}{Re}

\DeclareMathOperator{\diag}{diag}
\DeclareMathOperator{\C}{\mathbb{C}}
\DeclareMathOperator{\R}{\mathbb{R}}

\def\norm#1{ {|\hspace{-.022in}|#1|\hspace{-.022in}|} }
\newcommand{\spn}{\mathrm{span}}

\def\abs#1{\left|#1\right|}

\newcommand{\nc}{\newcommand}

\nc\bR{\mathbb{R}}
\nc\rH{\mathrm{H}}
\nc\N{\mathbb{N}}
 \newcommand{\trans}{^\top}

\def\>{\rangle}
\def\<{\langle}

\def\al{\alpha}

\def\be{\beta}

\def\ga{\gamma}

\nc{\ep}{\epsilon}
\def\eps{\epsilon}

\def\x{\xi}

\newtheorem{theorem}{Theorem}
\newtheorem{lemma}{Lemma}
\newtheorem{corollary}{Corollary}

\newtheorem{proposition}{Proposition}

\begin{document}

\title{Local extrema of entropy functions under tensor products}

\author{Shmuel Friedland}\email{friedlan@uic.edu}
\affiliation{Department of Mathematics, Statistics and Computer Science
University of Illinois at Chicago,
851 S. Morgan Street,
Chicago, IL 60607-7045, USA}
\author{Gilad Gour}\email{gour@ucalgary.ca}
\affiliation{Institute for Quantum Information Science and
Department of Mathematics and Statistics,
University of Calgary, 2500 University Drive NW,
Calgary, Alberta, Canada T2N 1N4}
\author{Aidan Roy}\email{aproy@uwaterloo.ca}
\affiliation{Institute for Quantum Computing,
University of Waterloo, 200 University Avenue West,
Waterloo, Ontario, Canada N2L 3G}

\date{\today}

\begin{abstract}
We show that under a certain condition of local commutativity the minimum von-Neumann entropy output of a quantum channel is locally additive . We also show that local minima of the $2$-norm entropy functions are closed under tensor products if one of the subspaces has dimension $2$.
\end{abstract}

\maketitle

Let $K$ be a subspace of the $m \times n$ complex matrices, and let $x \in K$, $\Tr[xx^*] = 1$. Then the von Neumann entropy of $x$ is
\[
\rH(x) := -\Tr[xx^* \ln xx^*],
\]
and the minimum entropy output of the subspace $K$ is 
\[
\rH(K) := \min_{x\in K, \Tr[xx^*] = 1} \rH(xx^*).
\]
Recently, Hastings \cite{Hastings08} disproved the famous additivity conjecture, which posited that
\begin{equation}
\label{eqn:addconj}
\rH(K_1 \otimes K_2) = \rH(K_1) + \rH(K_2).
\end{equation}
This conjecture was considered one of the most significant open problems in quantum information theory, spawning a large literature \cite{Holevo06}. Its importance was motivated in part by the problem of finding the classical capacity of a quantum channel, and in part by a result of Peter Shor \cite{Shor04} that showed that a number of apparently distinct additivity conjectures, including the additivity of the minimum entropy output of a quantum channel, the additivity of the entanglement of formation, and the additivity of the Holevo capacity, were all equivalent. 

Hastings' counterexample showed that the von Neumann entropy function is not globally additive on subspaces: in other words, if $x_1$ is a global minimum in $K_1$ and $x_2$ is a global minimum in $K_2$, then $x_1 \otimes x_2$ is not necessarily a global minimum in $K_1 \otimes K_2$. On the other hand, in this paper we show that under certain conditions the von Neumann entropy is locally additive. More precisely, we show that if $K_i$ is a subspace with a local minimum $x_i$, and $x_ix_i^*$ commutes with $x_iy_i^*$ for every $y_i \in K_i$, then $x_1 \otimes x_2$ is a local minimum of $K_1 \otimes K_2$; we call this condition the local commutativity condition. More generally, we study the behaviour of entropy functions of the eigenvalues of $xx^*$, and we consider when the tensor product of two local minima is again a local minimum. 

The paper is organized as follows. In Section~\ref{sec:lc} we analyze the local commutativity condition. In Section \ref{sec:firstderiv}, we consider the first derivative of the entropy function and note that critical points of the von Neumann and Renyi entropies are closed under tensor products. These results are due to a group participating in the American Institute for Mathematics workshop on ``Geometry and representation theory"\cite{AIM08}. 
In Section~\ref{sec:shannon}, we consider the second derivative of the von Neumann entropy function, and show that local minima of von Neumann entropy are closed under tensor products, given the previously mentioned commutativity assumption.  Finally, in Section \ref{sec:2norm}, we consider the second derivative of the $2$-norm entropy function. We show that local minima of the $2$-norm are closed under tensor products if one of the subspaces has dimension $2$. 
In the Appendix A we analyze the affine parametrization and use it to derive a necessary condition for local minima. In Appendix B we show that there is a simple counter example for the additivity conjecture over the real numbers.

\section{The local commutativity condition}\label{sec:lc}

For a given function $f:[0,\infty)\to (-\infty,\infty)$ we define $f(x)=\sum_{i=1}^m f(\lambda_i(xx^*))$ for $x\in \mathbb{C}^{m\times n}$, and $\lambda_i$ are the eigenvalues of $xx^*$.
We assume that either $f$ is smooth on $[0,\infty)$, i.e. has two continuous derivatives at every $t\ge 0$, or $f(t)=H(t)\equiv -t\log t$.
Let $D_yf(x), D_y^2 f(y)$ denote the first and the second derivative of $f$ in the $y$ direction:
\[
D_yf(x) = \frac{d}{d\ep}f(x + \ep y)\big|_{\ep = 0}, \quad D_y^2f(x) = \frac{d^2}{d^2\ep}f(x + \ep y)\big|_{\ep = 0}
\]
Then $x$ is a critical point if and only if $D_y f(x) = 0$  for each $y\in K$ (in the next section we will discuss in more details this condition).

Here we focus on the function $f(t)=H(t)\equiv -t\log t$.  In this case we need to be very careful when dealing with
$xx^*$ which have zero eigenvalues.  We will see that for any $x,y\in \C^{m\times n}$, $D_y f(x)\in \R$.  However it is possible that $D_y^2 f=\infty$,
and below we give the exact conditions on $y$ when this happens.
Hence if $x$ is a critical point of the von Neumann entropy, $H(x)$, and $D_y^2 H(x)=\infty$  then $H(x+\epsilon y)> H(x)$ for small enough $\epsilon$.
Thus when we study in the next sections the local minimum of $H(K_1\otimes K_2)$ at the critical point $x_1\otimes x_2$ we need only to consider $y_i$ such that
$D_{y_i}^2 f<\infty$ for $i=1,2$.  This will also give a partial explanation of the local commutativity condition discussed in the introduction.

\begin{lemma}\label{lemdyf=inf}  Let $x,y\in \C^{m\times n}$, $\Tr xx^*>0$ and $H(t)=-t\log t$.  Then  $D_yH(x)\in \R$.
Change standard orthonormal bases in $\C^{m},\C^n$ to new orthonormal bases such that $x,y$ have the forms
\begin{equation}\label{xypartition}
x=\left[\begin{array}{cc} x_{11}&0_{r,n-r}\\0_{m-r,r}&0_{m-r,n-r}\end{array}\right]\;\text{ and  }\;\;
y=\left[\begin{array}{cc} y_{11}&y_{12}\\y_{21}&y_{22}\end{array}\right], 
\end{equation}
with $0_{i,j} \in \C^{i\times j}$ and $x_{11}, y_{11}\in \C^{r\times r}$. Then $D_y^2 f(x) = \infty$ if and only if $y_{22}\ne 0$.
\end{lemma}
\proof By considering $UKV$, where $U,V$ unitary we may assume that $x,y$ in the form \eqref{xypartition}.  Furthermore $x_{11}=D\equiv\diag(d_1,\ldots,d_r) $, where $d_1\ge d_2\ge \ldots\ge d_r>0$ and $r$ is the rank of $x$.   So $d_i$ is the $i$-th singular value, $\sigma_i(x_{11})$ for $i=1,\ldots,r$.
Observe next that
\begin{equation}\label{tracexepy}
\Tr ((x+\ep y)(x+\ep y)^*)=\sum_{i=1}^m \lambda_i ((x+\ep y)(x+\ep y)^*).
 \end{equation}
We assume here that the eigenvalues of a hermitian matrix are arranged in a nonincreasing order. Note that
 \[(x+\ep y)(x+\ep y)^*=xx^*+\ep(xy^*+yx^*) +\ep^2(yy^*)\]
 Hence
 \[\lambda_i((x+\ep y)(x+\ep y)^*)=\lambda_i(xx^*+\ep(xy^*+yx^*)) +O(\ep^2).\]
 Observe next that
 \[xx^*+\ep(xy^*+yx^*)=\left[\begin{array}{cc} D+\ep(Dy_{11}^*+y_{11}D)&\ep Dy_{21}^*\\\ep y_{21}D&0\end{array}\right].\]
For small $\ep$, the first variation formula (see \cite{Kato}) yields
 \begin{align*}
& \lambda_i((x+\ep y)(x+\ep y)^*)= d_i+d_i'\ep   + O(\ep^2) \textrm{ for }i=1,\ldots, r, \\
 & \lambda_i((x+\ep y)(x+\ep y)^*)= O(\ep^2) \textrm{ for } i>r.
 \end{align*}
 Hence $\lambda_i((x+\ep y)(x+\ep y)^*)=d_i''\ep^2 +O(\ep^3)$ for $i > r$, with $d_{r+1}''\ge\ldots\ge d_m''\ge 0$. These calculations show that
 \[
 H(x+\ep y)=H(D+\ep y_{11})-\sum_{i=r+1}^m d_i''\ep^2\log(d_i''\ep^2) +O(\ep^2).
\]
 Hence $D_yH(x)\in\R$ and $D_y^2H(x)=\infty$ if and only if $d_{r+1}''>0$.
 It is left to show that $d_{r+1}''>0$ if and only if $y_{22}\ne 0$.
 Consider $\wedge^{r+1} (x+\ep y)$. the $r+1$ compound matrix of $x+\ep y$. 
(Recall that $\wedge^{r+1} (x+\ep y)$ is the ${m\choose r+1}\times {n\choose r+1}$ matrix whose entries are $(r+1)\times (r+1)$
 minors of $x+\ep y$.)  Note that $\wedge^{r+1} (x+\ep y)$ is a polynomial matrix in $\ep$.  Since $x$ has rank $r$ it follows that
 $\wedge^{r+1}(x)=0$.  Hence $\wedge^{r+1} (x+\ep y)=\ep z_1 +\ep^2 z_2(\ep)$,  where $z_1$ is a constant matrix and $z_2(\ep)$ is a polynomial
 matrix in $\ep$.  We claim that $z_1=0$ if and only if $y_{22}=0$.  Indeed since $D$ is diagonal then a minor of order $r+1$ that can have a nonzero derivative at $\ep=0$ is the minor based on the rows $\alpha=\{1,\ldots,r,p\}$ and columns $\beta=\{1,\ldots,r,q\}$.
 Denote this minor by $\det (x+\ep y)[\alpha,\beta]$.  Clearly  $\det (x+\ep y)[\alpha,\beta]=\ep (d_1\ldots d_r y_{p,q})+O(\ep^2)$, where $y_{p,q}$
 is the $p,q$ entry of $y$.  So if $y_{22}=0$ we obtain that $z_1=0$.  Hence $\|\wedge^{r+1}(x+\ep y)\|_2=\sigma_1(\wedge^{r+1}(x+\ep y))\le \ep^2 a$ for some positive $a$. Recall that 
 \[(\sigma_1(\wedge^{r+1}(\x+\ep y)))^2= \prod_{i=1}^{r+1} \lambda_i((x+\ep y)(x+\ep y)^*).\]
 As $(\sigma_1(\wedge^{r+1}(\x+\ep y) ))^2\le a^2 \ep^4$, we deduce that $d_{r+1}''=0$.

 It is left to show that if $y_{p,q}\ne 0$ for some $p,q>r$, then $d_{r+1}''>0$.  Clearly,
 \[\|\wedge^{r+1}(x+\ep y)\|_2\ge |\det(x+\ep y)[\alpha,\beta]|\ge d_1\ldots d_r|y_{p,q}| \frac{|\ep|}{2}\]
 for some small value of $\ep$.  (The first inequality follows from the fact the $\ell_2$ norm of a matrix is not less than the absolute value
 of any of its entries.)  This shows that $d_{r+1}''>0$.  \qed

The lemma above implies that for the purpose of calculating local minima, without loss of generality, we can always take the directional derivatives in a direction with $y_{22}=0$. In the lemma above, however, we did not impose the normalization condition
$\text{Tr}(xx^*)=1$. As we show in the next lemma, it does not affect the result that $D^2_y H=\infty$ if and only if $y_{22}=0$. 

 \begin{lemma}\label{varxyepsilon}
 Let $x,y\in \C^{m\times n}$, with $\Tr(xx^*)=1$ and $y\ne 0$. Consider the matrix 
 $$
 x(y,\ep):=\frac{1}{\sqrt{\Tr((x+\ep y)(x+\ep y)^*)}} (x+\ep y)\;,
 $$
 which is always defined for small $|\ep|$.   Then $\frac{d}{d\ep} H(x(y,\ep))|_{\ep=0}\in \R$, and  
 $\frac{d^2}{d\ep^2} H(x(y,\ep))|_{\ep=0}=\infty$  if and only if $D^2_y(f)=\infty$.
 \end{lemma} 
 \proof  
  The functions $h_1(\ep):=(\Tr((x+\ep y)(x+\ep y)^*))^{-1}$ and $h_2(\ep):=\log \Tr((x+\ep y)(x+\ep y)^*)$ are analytic in the neighborhood
 of $\ep=0$, and clearly
 \begin{equation}\label{hxyepexpr}
 H(x(y,\ep))=h_1(\ep)f(\ep)+h_2(\ep).
 \end{equation}
 As $h_1(0)=1$ we obtain
 \[
 \frac{d}{d\ep} H(x(y,\ep))|_{\ep=0}=D_y(f)+h_1'(0)H(x)+h_2'(0)\in\R,
 \]
 while $\frac{d^2}{d\ep^2} H(x(y,\ep))|_{\ep=0}$ consists of $D_y^2(f)$ plus finite terms. The lemma follows.  \qed
 
 The two lemmas above imply the following characterization of the local commutative condition discussed in the introduction.

 \begin{lemma}\label{loccomcon}  Let the assumptions of Lemma \ref{lemdyf=inf}  hold.  Assume that $x,y$ are in the form \eqref{xypartition}.
 Then $xx^*$ commutes with $xy^*$ if and only if $y_{21}=0$ and $x_{11}x_{11}^*$ commutes with $x_{11}y_{11}^*$ (which is equivalent to $x_{11}^*x_{11}y_{11}^*=y_{11}^*x_{11}x_{11}^*$.) 
 \end{lemma}
 \proof Write $x$ and $y$ as in \eqref{xypartition}.  The assumption that $x_{11}$ is invertible, and $xx^*$ commutes with $xy^*$ 
 is equivalent to $y_{21}=0$ and $x_{11}x_{11}^*$ commutes with $x_{11}y_{11}^*$.  So $x_{11}x_{11}^* x_{11}y_{11}^* =x_{11}y_{11}^* x_{11}x_{11}^*$.
 Divide both sides of this equalities  by $x_{11}$ to obtain the lemma.   \qed

In particular, the above lemma together with the theorem in Section~\ref{sec:shannon} imply that local additivity holds 
for subspaces consisting of matrices $y$ as in Eq.~(\ref{xypartition}), with 
$y_{11}$ diagonal, $y_{21}=0$, and $y_{12}$ arbitrary.

\section{First derivative of entropy functions under tensor products}\label{sec:firstderiv}

All of the results in this section are due to the ``Quantum Information Group" participating in the workshop ``Geometry and representation theory", held at the American Institute for Mathematics~\cite{AIM08}; we record the results here for completeness.

For a given function $f(t)$ as defined above, let $D_yf(x)$ denote the derivative of $f$ in the $y$ direction: 
\[
D_yf(x) = \frac{d}{d\ep}f(x + \ep y)\big|_{\ep = 0}.
\]
Then $x$ is a critical point if and only if $D_y f(x) = 0$ for every $y$. Since we are interested in local minima in $K$ subject to $\Tr[xx^*] = 1$, we restrict $y$ to the tangent space $\{y \in K: D_y\Tr[xx^*] = 0\} = \{y \in K: \Tr[xy^* + yx^*] = 0\}$. Also, we restrict our attention to functions $f(x)$ which depend only on $xx^*$. Since $xx^*$ is invariant under $x \mapsto ix$, we may ignore $y = ix$. That is, $x \in K$ is critical if and only if $D_yf(x) = 0$ for every $y$ in the orthogonal subspace
\[
x^\perp := \{y \in K : \Tr[xy^*] = 0\}.
\]

Under tensor products, the orthogonal subspace has the following decomposition:
\[
(x_1 \otimes x_2)^\perp \;= \;\langle x_1 \rangle \otimes x_2^\perp \;\; \oplus \;\; x_1^\perp \otimes \langle x_2 \rangle \;\; \oplus \;\; x_1^\perp \otimes x_2^\perp.
\]

For a function $f(x)$ depending only on $xx^*$, a point $x \in K$ is critical in $K$ if and only if $D_yf(x) = 0$ for every $y \in x^\perp$. In general, given a univariate differentiable function $F$, a Taylor series expansion of $F$ shows that the matrix function $a \mapsto \Tr[F(a)]$ has directional derivative
\[
\frac{d}{d\ep}\Tr[F(a + \ep b)]\big|_{\ep = 0} = \Tr[F'(a)b].
\]
We are interested in the case $a = xx^*$ and $b = xy^*+yx^*$: if $f(x) = \Tr[F(xx^*)]$, then
\[
D_yf(x) = \Tr[F'(xx^*)(xy^* + yx^*)].
\]
This derivative is zero for all $y \in x^\perp$ if and only if $\Tr[F'(xx^*)xy^*] = 0$ for all $y \in x^\perp$.

\begin{theorem}
\label{thm:firstderiv}
Let $F$ be a differentiable univariate function such that $F'(a_1 \otimes a_2)$ is in the span of 
\[
\{F'(a_1)\otimes F'(a_2),F'(a_1)\otimes I,I \otimes F'(a_2),I \otimes I\}.
\]
If $x_1$ and $x_2$ are critical points of $f(x) = \Tr[F(xx^*)]$ subject to $\Tr[xx^*] = 1$, then so is $x_1 \otimes x_2$. 
\end{theorem}

\begin{proof}
Let $x = x_1 \otimes x_2$. It suffices to show that if $D_{y_i}f(x_i) = 0$ for all $y_i \in x_i^\perp$, then $D_{y}f(x) = 0$ for all $y \in x^\perp$. That is, if $\Tr[F'(x_ix_i^*)x_iy_i^*] = 0$, then $\Tr[F'(xx^*)xy^*] = 0$. 

First, suppose $y = y_1 \otimes y_2$, for some arbitrary $y_1$ and $y_2$, and consider the term in $F(xx^*)$ proportional to $F'(x_1x_1^*)\otimes F'(x_2x_2^*)$: we have
\begin{align*}
\Tr[\left(F'(x_1x_1^*)\otimes F'(x_2x_2^*)\right)(xy^*)] \qquad \qquad\\
 = \Tr[F'(x_1x_1^*)x_1y_1^*]\Tr[F'(x_2x_2^*)x_2y_2^*],
\end{align*}
which is $0$ provided that either $y_1 \in x_1^\perp$ or $y_2 \in x_2^\perp$ (or both). Likewise, for the term proportional to $F'(x_1x_1^*)\otimes I$, 
\[
\Tr[\left(F'(x_1x_1^*)\otimes I\right)(xy^*)] = \Tr[F'(x_1x_1^*)x_1y_1^*]\Tr[x_2y_2^*],
\]
which again is $0$ if either $y_1 \in x_1^\perp$ or $y_2 \in x_2^\perp$. Similarly, $\Tr[\left(I\otimes F'(x_2x_2^*)\right)(xy^*)] = 0$ and $\Tr[\left(I\otimes I\right)(xy^*)] = 0$. Combining the terms which make up $F'(xx^*)$, we see that $\Tr[F'(xx^*)(xy^*)] = 0$ whenever $y = y_1 \otimes y_2$ satisfies $y_1 \in x_1^\perp$ or $y_2 \in x_2^\perp$.

Now an arbitrary element $y \in x^\perp$ can be written as a linear combination of terms of the form $x_1 \otimes y_2$, $y_1 \otimes x_2$, and $y_1 \otimes y_2$, with $y_i \in x_i^\perp$. For each of these terms either the first or second component of the tensor product is in $x_i^\perp$. Therefore $\Tr[F'(xx^*)xy^*] = 0$ for all $y \in x^\perp$.
\end{proof}

Our main interest is in the function $x \mapsto -\Tr[xx^* \ln xx^*]$, which is proportional to the usual von Neumann entropy of the matrix $xx^*$. Letting $F(t) = -t \ln t$, so that $F'(t) = -(1 + \ln t)$, we have
\begin{align*}
F'(a_1 \otimes a_2) & = -I - \ln(a_1 \otimes a_2) \\
& = -I \otimes I - \ln(a_1) \otimes I - I \otimes \ln(a_2) \\
& \in \spn\left\{I \otimes I,F'(a_1)\otimes I, I \otimes F'(a_2)\right\}.
\end{align*}
(Here we used the fact that $\ln(a_1 \otimes a_2) = \ln(a_1) \otimes I + I \otimes \ln(a_2)$.) Thus the hypotheses of Theorem \ref{thm:firstderiv} are satisfied, and so critical points of $x \mapsto -\Tr[xx^* \ln xx^*]$ are closed under tensor products.
 
Another important class of entropy functions are the $p$-norms:
\[
x \mapsto \norm{xx^*}_p^p = \Tr[(xx^*)^p].
\]
Letting $F(t) = t^p$, so $F'(t) = pt^{p-1}$, we have 
\[
F'(a_1 \otimes a_2) = p(a_1 \otimes a_2)^{p-1} = \frac{1}{p}F'(a_1) \otimes F'(a_2).
\] 
Again $F(t)$ is in the form of Theorem \ref{thm:firstderiv}. Thus for both the von Neumann entropy and the $p$-norms, the tensor product of critical points (subject to $\Tr[xx^*] = 1$) are again critical points.

\section{Second derivative of the von-Neumann entropy}\label{sec:shannon}

In this section we show that under the local commutativity condition, if $x_1 \in K_1$ and $x_2 \in K_2$ are nonsingular strong local minima of 
\[
x \mapsto -\Tr[xx^* \log xx^*]
\]
subject to $\Tr[xx^*] = 1$, then $x_1 \otimes x_2$ is also a strong local minimum in $K_1 \otimes K_2$. More precisely, we assume that if $y_i \in K_i$ is orthogonal to $x_i$, then $x_ix_i^*$ and $x_iy_i^*$ commute. Throughout this section we will also assume without loss of generality that $\Tr[yy^*] = 1$.

In this section we work with the normalized entropy function
\[
\rH(x) := -\Tr\bigg[\frac{xx^*}{\norm{x}^2} \log \frac{xx^*}{\norm{x}^2}\bigg].
\]
A point $x$ is a strong local minimum of $\rH$ on $\{x : \Tr[xx^*] = 1\}$ if and only if for every $y$ orthogonal to $x$, the second directional derivative $D_y^2\rH(x)$ is positive.

\begin{lemma} 
\label{lem:secondderiv}
Assume $xx^*$ and $xy^*$ commute. Then 
\begin{eqnarray*}
D^{2}_{y}\rH(x)=2{\rm Tr}\left[xx^* \log xx^*\right]-2{\rm Tr}\left[yy^*\log xx^*\right] \\
-{\rm Tr}[(xy^* + yx^*)^2(xx^*)^{-1}]\;,
\end{eqnarray*}
where the last trace is taken over the support of $xx^*$.
\end{lemma}

\begin{proof}
For convenience define $a = xx^*$, $b = xy^* + yx^*$, and $c = yy^*$, so that
\[
(x+ \eps y)(x+\eps y)^* = a + \eps b + \eps^2c.
\]
Note that $\Tr[a] = \Tr[c] = 1$ and $\Tr[b] = 0$, so $\Tr[(x+ \eps y)(x+\eps y)^*] = 1 + \eps^2$. Then 
\begin{align*}
& \rH(x + \eps y) = -\Tr \left[\frac{a + \eps b + \eps^2c}{1 + \eps^2} \log \frac{a + \eps b + \eps^2c}{1 + \eps^2} \right] \\
& = -\frac{\Tr[(a + \eps b + \eps^2c) \log (a + \eps b + \eps^2c)]}{1 + \eps^2} + \log (1 + \eps^2).
\end{align*}

Up to a second order in $\varepsilon$ this expression becomes
\begin{align*}
& \rH(x + \eps y) =-{\rm Tr}\left[a\log\left(a+\varepsilon b+\varepsilon ^2 c\right)\right] \\
& \; -\varepsilon{\rm Tr}\left[b\log\left(a+\varepsilon b\right)\right]+\varepsilon^2\big(1+{\rm Tr}\left[a \log a\right]-{\rm Tr}\left[c\log a\right]\big).
\end{align*}
Therefore, the second order directional derivative can be expressed in the following way:
\begin{align*}
& D^{2}_{y}\rH(x)=-\frac{d^2}{d\varepsilon^2}{\rm Tr}\left[a\log\left(a+\varepsilon b+\varepsilon ^2 c\right)\right] \Big|_{\varepsilon=0} - \\
& 2\frac{d}{d\varepsilon}{\rm Tr}\left[b\log\left(a+\varepsilon b\right)\right]\Big|_{\varepsilon=0}\!+\!2\big(1+{\rm Tr}\left[a \log a\right]-{\rm Tr}\left[c\log a\right]\big).
\end{align*}
To calculate the derivative expressions above, we will express the log function by its Taylor series:
$$
\log(a+\varepsilon b)=\log[I-(I-a-\varepsilon b)]=-\sum_{n=1}^{\infty}\frac{(I-a -\varepsilon b)^n}{n}.
$$
Without loss of generality (see Lemma~\ref{lemdyf=inf}), in the last equality we assumed that $a$ is invertible, so that for sufficiently small $\varepsilon$
also $a+\varepsilon b$ is invertible and therefore $I-a-\varepsilon b<I$. To calculate the derivative of ${\rm Tr}\left[b\log\left(a+\varepsilon b\right)\right]$ at $\varepsilon = 0$, we only need to take terms proportional to $\varepsilon$ in the expansion of the logarithm. Assuming $a$ and $b$ commute,
\begin{align*}
\frac{d}{d\varepsilon}{\rm Tr}\left[b\log\left(a+\varepsilon b\right)\right]\Big|_{\varepsilon=0} & =\sum_{n=1}^{\infty}{\rm Tr}\left[b^2(I-a)^{n-1}\right] \\
& = {\rm Tr}\left[b^2a^{-1}\right]. 
\end{align*}
To calculate the second derivative of ${\rm Tr}[a\log(a+\varepsilon b+\varepsilon ^2 c)]$ we need only take the terms proportional to $\varepsilon^2$. Again assuming $a$ and $b$ commute and $a$ is invertible,
\begin{align*}
& \frac{d^2}{d\varepsilon^2}{\rm Tr}\left[a\log\left(a+\varepsilon b+\varepsilon ^2 c\right)\right]
\Big|_{\varepsilon=0} \\
&=-\sum_{n=2}^{\infty}\frac{2}{n}\binom{n}{2}{\rm Tr}\left[a(I-a)^{n-2}b^2\right]+2\sum_{n=1}^{\infty}{\rm Tr}\left[a(I-a)^{n-1}c\right] \\
&= -{\rm Tr}\left[a^{-1}b^2\right]+2\Tr[c]  =-{\rm Tr}\left[a^{-1}b^2\right]+2.
\end{align*}

Therefore $D^{2}_{y}\rH(x)=2{\rm Tr}\left[a \log a\right]-2{\rm Tr}\left[c\log a\right]-{\rm Tr}[b^2a^{-1}]$.
\end{proof}

\begin{corollary} 
\label{cor:secondderiv}
Assume $xx^*$ and $xy^*$ commute. Then $D^2_y \rH(x) > 0$ if and only if 
\begin{align*}
\abs{\Tr[(xx^*)^{-1}(xy^*)^2]} + \Tr[(xx^*)^{-1}xy^*yx^*]  \\
< \Tr[xx^* \log xx^*] - \Tr[yy^* \log xx^*]\;,
\end{align*}
where $(xx^{*})^{-1}$ is the inverse over the support of $xx^*$.
\end{corollary}

\begin{proof}
Expand $(xy^*+yx^*)^2$ into four terms, noting that $xy^*$ and $yx^*$ commute with $(xx^*)^{-1}$. Then $D_y\rH(x) > 0$ if and only if
\begin{align*}
& \Tr[(xx^*)^{-1}(xy^*)^2] + \Tr[(xx^*)^{-1}+(yx^*)^2] + 2\Tr[(xx^*)^{-1}xy^*yx^*] \\
& < -2\Tr[yy^* \log xx^*]  + 2\Tr[xx^* \log xx^*].
\end{align*}
The first two terms on the LHS are twice the real part of $\Tr[(xx^*)^{-1}(xy^*)^2]$; the largest value of these two terms over all phases of $y$ is $2\abs{\Tr[(xx^*)^{-1}(xy^*)^2]}$.
\end{proof}

For convenience, denote the terms in Corollary \ref{cor:secondderiv} as follows:
\begin{align}
a(x,y) & := \abs{\Tr[(xx^*)^{-1}(xy^*)^2]}, \notag \\
b(x,y) & := \Tr[(xx^*)^{-1}xy^*yx^*], \label{eqn:abcd} \\
c(x) & := \Tr[xx^* \log xx^*],  \notag \\
d(x,y) & := \Tr[yy^* \log xx^*], \notag 
\end{align}
so $D^2_y \rH(x) > 0$ if and only if $a + b < c - d$. Each of these terms behaves nicely under tensor products:
\begin{align}
a(x_1 \otimes x_2,y_1 \otimes y_2) & = a(x_1,y_1) a(x_2,y_2), \label{eqn:abcdtensor} \\
b(x_1 \otimes x_2,y_1 \otimes y_2) & = b(x_1,y_1) b(x_2,y_2), \notag \\
c(x_1 \otimes x_2) & = c(x_1) + c(x_2),\notag \\
d(x_1 \otimes x_2,y_1 \otimes y_2) & = d(x_1,y_1) + d(x_2,y_2).\notag 
\end{align}
We can also bound the size of some of these terms for any $x$ and $y$ such that $\Tr[xy^*] = 0$ and $\Tr[xx^*] = \Tr[yy^*] = 1$. First, we claim $b \in [0,1]$. To see this, note that $P = x^*(xx^*)^{-1}x$ is a projection matrix,
so
\[
b = \Tr[y^*yx^*(xx^*)^{-1}x] = \norm{Py^*}^2,
\]
and $0 \leq \norm{Py^*}^2 \leq \norm{y^*} = 1$. Second, we claim $a \in [0,b]$. To see this, note that without loss of generality (see Lemma~\ref{lemdyf=inf}) we can assume that $xx^*$ is invertible and therefore positive definite, so $(xx^*)^{-1/2}$ exists and commutes with $xy^*$, and so $(xx^*)^{-1}(xy^*)^2 = ((xx^*)^{-1/2}xy^*)^2$. By Cauchy-Schwartz,
\begin{align*}
a & = \abs{\Tr[((xx^*)^{-1/2}xy^*)^2]} \\
& \leq \Tr[((xx^*)^{-1/2}xy^*)((xx^*)^{-1/2}xy^*)^*] = b.
\end{align*}
Thirdly, we claim that $c \leq 0$, since it is the negative of the entropy function. We are now ready to prove the main result
of this paper.

\begin{theorem}
Suppose $x_1$ and $x_2$ are strong local minima of $x \mapsto -\Tr[xx^* \log xx^*]$ subject to $\Tr[xx^*] = 1$ and $x_i \in K_i$, where $K_i$ is a subspace. Further assume that for every $y_i \in K_i$, the matrices $x_ix_i^*$ and $x_iy_i^*$ commute. Then $x := x_1 \otimes x_2$ is a strong local minimum in $K_1 \otimes K_2$.
\end{theorem}

\begin{proof}
We show that under the hypotheses of the theorem, if $D^2_{y_i} \rH(x_i)$ is positive for every $y_i \in x_i^\perp$, then $D^2_{y} \rH(x)$ is positive for every $y \in x^{\perp}$. We break the proof into several cases depending on $y$.

First, suppose $y$ is a tensor product. 

\underline{Case $y = x_1 \otimes y_2$, $y_2 \in x_2^\perp$:} Since $y_2 \in x_2^\perp$ and $x_2$ is a strong local minimum, we know that 
\[
a(x_2,y_2) + b(x_2,y_2) < c(x_2) - d(x_2,y_2).
\]
It is also easy to see from the expressions \eqref{eqn:abcd} that
\[
a(x_1,x_1) = b(x_1,x_1) = 1, \quad c(x_1) = d(x_1,x_1).
\]
So, using the expressions for tensors in \eqref{eqn:abcdtensor}, we have
\begin{align*}
a(x,y) + b(x,y) & = a(x_2,y_2) + b(x_2,y_2) \\
& \leq c(x_2) - d(x_2,y_2) \\
& = c(x) - d(x,y).
\end{align*}
Thus the second directional derivative is positive for this choice of $y$.

\underline{Case $y = y_1 \otimes x_2$, $y_1 \in x_1^\perp$:} This case is similar to $y = x_1 \otimes y_2$.

\underline{Case $y = y_1 \otimes y_2$, $y_i \in x_i^\perp$:} Here we require the arithmetic-geometric mean inequality. For two terms $a_1,a_2 \leq 1$,
\[
a_1a_2 \leq \Big(\frac{a_1+a_2}{2}\Big)^2  \leq \frac{1}{2}(a_1+a_2).
\]
In particular, $a(x_1,x_1)a(x_2,y_2) \leq a(x_1,y_1) + a(x_2,y_2)$ and similarly for $b$.
Now, since $y_i \in x_i^\perp$, we have $a(x_i,y_i) + b(x_i,y_i) < c(x_i) - d(x_i,y_i)$. Combining these inequalities we get $a(x,y) + b(x,y) \leq c(x) - d(x,y)$.

Next, we consider cases where $y$ is a linear combination of terms.

Suppose $y$ is in $x_1^\perp \otimes x_2^\perp$. In this case, we break $y$ into two orthogonal pieces according to the projection matrix $P = x^*(xx^*)^{-1}x$. Let $P_i = x_i^*(x_ix_i^*)^{-1}x_i$: this is the projection matrix onto the range of $x_i^*$, which we denote $R(x_i^*)$. Then $P = P_1 \otimes P_2$ is the projection matrix onto the range $R(x^*) = R(x_1^*) \otimes R(x_2^*)$. Write $y$ as a direct sum:
\[
y = \al u + \be v,
\]
where $u^* \in R(x^*)$ (so $Pu^* = u^*$), and $Pv^* = 0$. The normalizations are chosen so that $\al \in \bR$ and $\be \in \bR$ satisfy $\al^2 + \be^2 = 1$, and $\norm{u}^2 = \norm{v}^2 = 1$. We deal with the $u$ and $v$ components separately.

\underline{Case $y = u \in (x_1^\perp \otimes x_2^\perp) \cap R(x^*)$:} Here we have $b(x,u) = \norm{Pu^*}^2 = \norm{u^*}^2 = 1$. Note that if $y_i$ is in $x_i^\perp$, then 
\[
\Tr[x_iP_iy_i^*] = \Tr[x_ix_i^*(x_ix_i^*)^{-1}x_iy_i^*] = \Tr[x_iy_i^*] = 0,
\]
so $P_iy_i^*$ is also in $x_i^\perp$. If we write $u = \sum_j y_{1j} \otimes y_{2j}$ with $y_{ij} \in x_i^\perp$, so that
\[
u = Pu = \sum_j P_1y_{1j} \otimes P_2y_{2j},
\]
then $P_iy_{ij}$ is in $x_i^\perp \cap R(x_i^*)$, and it follows that $u$ is in $(x_1^\perp \cap R(x_1^*)) \otimes (x_2^\perp \cap R(x_2^*))$. Now perform a Schmidt decomposition of $u$ with respect to this tensor space: we get 
\[
u = \sum_j \al_j u_{1j} \otimes u_{2j},
\]
where $u_{ij} \in x_i^\perp \cap R(x_i^*)$, $\Tr[u_{ij}u_{ik}^*] = \delta_{jk}$, $\al_j \geq 0$, and $\sum_j \al_j^2 = 1$. Since $u_{ij}$ is in $R(x_i^*)$, we have $b(x_i,u_{ij}) = 1$. Since $u_{ij}$ is in $x_i^\perp$, we know 
\begin{equation}
\label{eqn:abcdu}
a(x_i,u_{ij}) + b(x_i,u_{ij}) \leq c(x_i) - d(x_i,u_{ij}),
\end{equation}
and also $0 \leq a(x_i,u_{ij}) \leq 1$. Under this decomposition, we also have
\begin{align}
d(x,u) 
 = \sum_j \al_j^2 (d(x_1,u_{1,j}) + d(x_2,u_{2,j})). \label{eqn:dxu}
\end{align}
Therefore, from \eqref{eqn:abcdu} and \eqref{eqn:dxu},
\begin{align*}
a(x,u) & + b(x,u) \\
& \leq 2b(x,u) \\
& = \sum_{j} \al_j^2 [b(x_1,u_{1,j}) + b(x_2,u_{1,2})] \\
& \leq \sum_{j} \al_j^2 [c(x_1) - d(x_1,u_{1,j})] + c(x_2) - d(x_2,u_{1,2}) \\
& = c(x) - d(x,u).
\end{align*}

\underline{Case $y = v \in x_1^\perp \otimes x_2^\perp$, $Pv^* = 0$:} We know that $0 \leq a(x,v) \leq b(x,v) = \norm{Pv^*} = 0$, and so $a(x,v) = b(x,v) = 0$. Perform a Schmidt decomposition of $v$ with respect to the space $x_1^\perp \otimes x_2^\perp$: 
\[
v = \sum_j \be_j v_{1j} \otimes v_{2j},
\]
where $v_{ij} \in x_i^\perp$, $\Tr[v_{ij}v_{ik}^*] = \delta_{jk}$ and $\sum_j \be_j^2 = 1$. Since $v_{ij}$ is in $x_i^\perp$, we have $0 \leq a(x_i,v_{ij}) \leq b(x_i,v_{ij})$ and 
\begin{equation}
\label{eqn:cdpos}
0 \leq a(x_i,v_{ij}) + b(x_i,v_{ij}) \leq c(x_i) - d(x_i,v_{ij}).
\end{equation}
It follows quickly that $a(x,v) + b(x,v) \leq c(x) - d(x,v)$.

Next we deal with a combination of $u$ and $v$.

\underline{Case $y \in x_1^\perp \otimes x_2^\perp$:} Write $y = \al u + \be v$, where $u^* \in R(x^*)$), $Pv^* = 0$, $\al^2 + \be^2 = 1$, and $\norm{u}^2 = \norm{v}^2 = 1$. Then since $uv^* = uPv^* = 0$, we have
\[
yy^* = \al^2 uu^* + \be^2 vv^*,
\]
from which it follows that 
\begin{align}
b(x,y) & = \al^2b(x,u) + \be^2b(x,v), \label{eqn:bxuyv} \\
d(x,y) & = \al^2d(x,u) + \be^2d(x,v). \label{eqn:dxuyv}
\end{align}
(In fact, $b(x,u) = 1$ and $b(x,v) = 0$.) Combining \eqref{eqn:bxuyv} and \eqref{eqn:dxuyv} with the results for $u$ and $v$ from the previous cases, we get
\begin{align*}
a(x,y) + b(x,y) & \leq 2b(x,y) \\
& = \al^22b(x,u) + \be^22b(x,v) \\
& \leq \al^2[c(x) - d(x,u)] + \be^2[c(x) - d(x,v)] \\
& = c(x) - d(x,y).
\end{align*}

Finally, we have the case where $y$ is an arbitrary element of $x^\perp$.

\underline{Case $y \in x^\perp$:} Here  $y$ may be written in the form
\[
y = \al x_1 \otimes y_2 + \be y_1 \otimes x_2 + \ga y',
\]
where $y_i \in x_i^\perp$ and $y' \in x_1^\perp \otimes x_2^\perp$, with real constants satifying $\al^2 + \be^2 + \ga^2 = 1$. Expanding out terms of $yy^*$ and simplifying, we find that most terms disappear under trace:
\begin{align}
d(x,y) & = \al^2 [c(x_1)+d(x_2,y_2)] + \be^2 [c(x_2)+d(x_1,y_1)] \notag \\
& \qquad + \ga^2d(x,y'), \label{eqn:dgen}\\
b(x,y) & = \al^2 b(x_2,y_2) + \be^2 b(x_1,y_1) + \ga^2b(x,y'), \label{eqn:bgen}\\
a(x,y) & = \Big| \al^2 \Tr[(x_2x_2^*)^{-1}(x_2y_2^*)^2] \\
+ & \be^2 \Tr[(x_1x_1^*)^{-1}(x_1y_1^*)^2] + \ga^2\Tr[(xx^*)^{-1}(x(y')^*)^2] \Big|. \notag
\end{align}
The expression for $d(x,y)$ requires the observation that $\Tr[x_iy_i^* \log x_ix_i^*] = 0$, because the first directional derivative of $D_{y_i}\rH(x_i)$ is $0$ when $x_i$ is a local minimum. The expression for $a(x,y)$ is bounded as follows:
\begin{align}
& a(x,y) = \Big| \al^2 \Tr[(x_2x_2^*)^{-1}(x_2y_2^*)^2] \notag \\
& + \be^2 \Tr[(x_1x_1^*)^{-1}(x_1y_1^*)^2] + \ga^2\Tr[(xx^*)^{-1}(x(y')^*)^2] \Big|  \notag \\
& \qquad \;\; \leq \al^2 \abs{\Tr[(x_2x_2^*)^{-1}(x_2y_2^*)^2]} \notag \\
& + \be^2 \abs{\Tr[(x_1x_1^*)^{-1}(x_1y_1^*)^2]} + \ga^2 \abs{\Tr[(xx^*)^{-1}(x(y')^*)^2]} \notag \\
& \qquad \;\; = \al^2a(x_2,y_2) + \be^2a(x_1,y_1) + \ga^2a(x,y'). \label{eqn:agen}
\end{align}
Combining \eqref{eqn:dgen}, \eqref{eqn:bgen} and \eqref{eqn:agen}, we get $a(x,y) + b(x,y) \leq c(x) - d(x,y)$.
\end{proof}

\section{The second derivative of the $2$-norm}
\label{sec:2norm}

In this section we focus on the 2-norm since its second directional derivative has an elegant analytical form.
We prove that if $K_1$ and $K_2$ are subspaces of matrices, at least one of which has dimension $2$, and $x_1 \in K_1$, $x_2 \in K_2$ are strong local maxima of the $2$-norm function
\[
x \mapsto \Tr[(xx^*)^2]
\]
subject to $\Tr[xx^*] = 1$, then $x_1 \otimes x_2$ is also a strong local maximum in $K_1 \otimes K_2$.
Since it is known that the $2$-norm is not globally additive, this result sheds some light on the possibility
that there exist functions that are locally additive while they are not globally additive.

We will work with the normalized function
\[
\rH_2(x) := \Tr\left[\bigg(\frac{xx^*}{\norm{x}^2}\bigg)^2\right] = \frac{\Tr[(xx^*)^2]}{[\Tr(xx^*)]^2}.
\]
As before, we consider $(x + \eps y)(x + \eps y)^* = xx^* + \eps(xy^* + yx^*) + \eps^2(yy^*)$, where $\Tr[xx^*] = \Tr[yy^*] = 1$ and $\Tr[xy^* + yx^*] = 0$. Noting that 
\begin{align*}
[(x + \eps y)(x + \eps y)^*]^2  = (xx^*)^2 + 2\eps xx^*(xy^* + yx^*) \\
 + \eps^2[2xx^*yy^* + (xy^* + yx^*)^2] + O(\eps^3),
\end{align*}
and that $\Tr[(x + \eps y)(x + \eps y)^*]^2 = (1+\eps^2)^2$, we have that up to second order in $\eps$,
\begin{align*}
\rH_2(x + \eps y) & = \Tr[(xx^*)^2] + 2\eps \Tr[xx^*(xy^* + yx^*)] \\
& + \eps^2\Tr[2xx^*yy^* + (xy^* + yx^*)^2 - 2(xx^*)^2].
\end{align*}
Then the first directional derivative of $\rH_2$ is
\[
D_y \rH_2(x) = 2\Tr[xx^*(xy^* + yx^*)].
\]
By considering $iy$ as well as $y$, the condition $D_y \rH_2(x) = 0$ reduces to $\Tr[xx^*xy^*] = 0$. The second derivative is 
\[
D^2_y \rH_2(x) = 2\Tr[(xy^* + yx^*)^2] + 4\Tr[xx^*yy^*] -4 \Tr[(xx^*)^2].
\]
For strong local maxima we expect $D^2_y \rH_2(x)$ to be negative. Expand $(xy^* + yx^*)^2$ into four terms: then
\[
\Tr[(xy^* + yx^*)^2] = 2 \bRp \Tr[(xy^*)^2] + \Tr[x^*xy^*y],
\]
where $\bRp(z)$ denotes the real part of $z$. The largest value of $\bRp \Tr[(xy^*)^2]$ over all choices of unit multiples of $y$ is $\abs{\Tr[(xy^*)^2]}$. In summary: 

\begin{lemma}
Define $F(x,y) := $
\[
-\Tr[(xx^*)^2] + \abs{\Tr[(xy^*)^2]} + \Tr[xx^*yy^*] + \Tr[x^*xy^*y].
\]
Then $x$ is a strong local maximum of the function $x \mapsto \Tr[(xx^*)^2]$, subject to $\Tr[xx^*] = 1$, if and only if every $y \in x^\perp$, $\Tr[yy^*] = 1$ satisfies $\Tr[xx^*xy^*] = 0$ and $F(x,y) < 0$.
\end{lemma}

Denote the terms in $F(x,y)$ as follows:
\begin{align*}
a(x) & := \Tr[(xx^*)^2] \\
b(x,y) & := \abs{\Tr[(xy^*)^2]} \\
c(x,y) & := \Tr[xx^*yy^*] \\
d(x,y) & := \Tr[x^*xy^*y].
\end{align*}
The Cauchy-Schwartz inequality implies that $\abs{\Tr[z^2]} \leq \Tr[z^*z]$ for any matrix $z$. Letting $z = xy^*$, we conclude that $0 \leq b(x,y) \leq c(x,y),d(x,y)$. Assuming $a> b+c+d$ therefore implies that $a > b,c,d$. Since $a = \Tr[(xx^*)^2)] \leq \Tr[xx^*]= 1$, we see that each of the terms $a,b,c,d$ are in the range $[0,1]$. Furthermore, each term is multiplicative under tensor products: $a(x_1 \otimes x_2)  = a(x_1) a(x_2)$, $b(x_1 \otimes x_2,y_1 \otimes y_2)  = b(x_1,y_1) b(x_2,y_2)$, and so on.

From Section \ref{sec:firstderiv}, we know that that tensor products of critical points of the $2$-norm are again critical points. We can now say the same for local maxima.

\begin{lemma}
Suppose $x_1$ and $x_2$ are strong local maxima of $x \mapsto \Tr[(xx^*)^2]$ subject to $\Tr[xx^*] = 1$ and $x_i \in K_i$, where either $K_1$ or $K_2$ has dimension $2$. Then $x := x_1 \otimes x_2$ is a strong local maximum in $K_1 \otimes K_2$.
\end{lemma}

\begin{proof}
Without loss of generality, assume $K_1$ has dimension $2$, so $x_1^\perp$ has dimension $1$. Let $y_1$ be an element of $x_1^\perp$ and let $y_{2j}$ be elements of $x_2^\perp$. Then every element of $x^\perp$ in $K_1 \otimes K_2$ is a linear combination of vectors of the form $y_1 \otimes y_{21}$, $x_1 \otimes y_{22}$, and $y_1 \otimes x_2$. First, we check that for each $y$ of that form, $F(x,y)$ is negative. 

\underline{Case $y = y_1 \otimes y_{21}$:} Here 
\begin{align*}
F(x,y) = - a(x_1)a(x_2) + b(x_1,y_1)b(x_2,y_{21}) \\
+ c(x_1,y_1)c(x_2,y_{21}) + d(x_1,y_1)d(x_2,y_{21}).
\end{align*}
Since $a(x_i) > b(x_i,y_i) + c(x_i,y_i) + d(x_i,y_i)$ and each term is nonnegative, it follows that 
$a(x_1)a(x_2) > b(x_1,y_1)b(x_2,y_{21}) + c(x_1,y_1)c(x_2,y_{21}) + d(x_1,y_1)d(x_2,y_{21})$,
and so $F(x,y)$ is negative.

\underline{Case $y = x_1 \otimes y_{22}$:} Here 
\begin{align*}
F(x,y) = - a(x_1)a(x_2) + a(x_1)b(x_2,y_{22}) \\
+ a(x_1)c(x_2,y_{22}) + a(x_1)d(x_2,y_{22}). 
\end{align*}
But $-a(x_2) + b(x_2,y_{22}) + c(x_2,y_{22}) + d(x_2,y_{22}) < 0$ and $a(x_1) > 0$, so $F(x,y)$ is negative. 

\underline{Case $y = y_1 \otimes x_2$:} Similar to $y = x_1 \otimes y_{22}$.

Now consider a linear combination of the three elements of $x^\perp$, say
\[
y = \al(y_1 \otimes y_{21}) + \be(x_1 \otimes y_{22}) + \ga(y_1 \otimes x_2).
\]
In considering $b(x,y)$, most terms disappear under trace:
\begin{align*}
& b(x,y)  = \Big|\al^2 \Tr[(x_1y_1^*)^2]\Tr[(x_2y_{21}^*)^2] \\
& + \be^2\Tr[(x_1x_1^*)^2]\Tr[(x_2y_{22}^*)^2] + \ga^2\Tr[(x_1y_1^*)^2]\Tr[(x_2x_2^*)^2] \Big| \\
& \qquad \;\; \leq \abs{\al}^2b(x_1,y_1)b(x_2,y_{21}) \\
&  + \abs{\be}^2a(x_1)b(x_2,y_{22}) + \abs{\ga}^2b(x_1,y_1)a(x_2).
\end{align*}
Likewise we have 
\begin{align*}
c(x,y) & = \abs{\al}^2c(x_1,y_1)c(x_2,y_{21}) + \\
& \abs{\be}^2a(x_1)c(x_2,y_{22}) + \abs{\ga}^2c(x_1,y_1)a(x_2),
\end{align*}
and similarly for $d(x,y)$. Adding together, we conclude that 
\begin{align*}
F(x,y) \leq  \abs{\al}^2\big[\neg a(x_1)a(x_2)+b(x_1,y_1)b(x_2,y_{21})\\
 + c(x_1,y_1)c(x_2,y_{21})+d(x_1,y_1)d(x_2,y_{21})\big] \\ 
 +  \abs{\be}^2a(x_1) \left[-a(x_2)+b(x_2,y_{22})+c(x_2,y_{22})+d(x_2,y_{22})\right] \\
 + \abs{\ga}^2 \left[-a(x_1)+b(x_1,y_1)+c(x_1,y_1)+d(x_1,y_1)\right]a(x_2).
\end{align*}
The $\abs{\al}^2$ term is negative by the argument given in the case $y = y_1 \otimes y_{21}$; the $\abs{\be}^2$ term is negative by the case $y = x_1 \otimes y_{22}$; and the $\abs{\ga}^2$ term is negative by the case $y = y_1 \otimes x_2$.
\end{proof}
 
If both $K_1$ and $K_2$ have dimension higher than $2$, the linear combinations seem to be more difficult.

\section*{Acknowledgment}
We appreciate many valuable discussions with Jon Yard. 
Research by GG and AR was supported by NSERC, PIMS, and iCORE.

\appendix
\section{Affine parametrization}
Let $x\in K\subset \C^{m\times n}$ and assume that $xx^*>0$, i.e. $xx^*$ invertible, and $\Tr xx^*=1$.
When now study the second variation.  Let $y\in K$ and assume that $\Tr(xy^*+yx^*)=0$.  We then consider
\begin{align} 
& A(\ep):=A+\ep B +\ep^2 C=(x+\ep y)(x+\ep y)^*, \nonumber\\
 & A=x x^*,B =x y^* +y x^*, C=yy^*\in \rH_m.\label{quadpar}
\end{align}
(Here $\rH_m$ is the real space of $m\times m$ matrices.)
Let $\lambda_1(\ep),\ldots,\lambda_m(\ep)>0$ be the eigenvalues of $A(\ep)$, as analytic functions of $\ep$, (Rellich's theorem~\cite{Kato}).   
We can assume that these eigenvalues are arranged in the following order 
$\lambda_1(\ep)\ge\ldots\ge\lambda_m(\ep)>0$ for small \textbf{positive} $\ep$.
Let $A_1(\ep)=A+\ep B$.  Arrange the analytic eigenvalues of $A_1(\ep)$ in the order $\mu_1(\ep)\ge \ldots\ge \mu_m(\ep)>0$ for small
positive $\ep$.  Clearly, $\lambda_i(\ep)=\mu_i(\ep) +O(\ep^2)$ for $i=1,\ldots,m$.
The following result is known, and can be deduced from the arguments in Kato~\cite{Kato}.
\begin{lemma}\label{secpertlemma}  Let $A,B,C\in \rH_m$, and denote $A(\ep)=A+\ep B +\ep^2 C, A_1(\ep)=A+\ep B$.
Assume that $\lambda_1(\ep),\ldots,\lambda_m(\ep)$ and $\mu_1(\ep), \ldots, \mu_m(\ep)$ are analytic eigenvalues of $A(\ep),A_1(\ep)$
arranged in a nonincreasing order for small positive $\ep$.  Then, there exists a unitary matrix $U\in \C^{m\times m}$ with the following two properties. First, $UAU^*=\diag(\lambda_1,\ldots,\lambda_m)$. Second, if we denote
$UCU^*\equiv F=[f_{ij}]_{i=j=1}^m$ then
\begin{equation}\label{secperteigf}
\lambda_i(\ep)=\mu_i(\ep)+\ep^2 f_{ii} +O(\ep^3) \textrm{ for } i=1,\ldots,m.
\end{equation}
\end{lemma} 
In the next proposition we use the above lemma to calculate the variation of $S(A)\equiv-\Tr(A\log A)$ up to second order.
\begin{proposition}\label{affparfor}
Let $x,y\in \C^{m\times n}$ and assume that $\Tr(xx^*)=\Tr(yy^*)=1, xx^*>0,\Tr(xy^*+yx^*)=0$.  Define $A(\ep),A_1(\ep)$ as in Eq.~\eqref{quadpar}.
Then
\begin{equation}\label{affparfor1}
S\left(\frac{A(\ep)}{\Tr A(\ep)}\right)=S(A_1(\ep))+\ep^2\Tr\left[(xx^*-yy^*) \log xx^*\right]
+O(\ep^3)
\end{equation}
\end{proposition}
\proof  First recall that
\begin{align}
&S\left(\frac{A(\ep)}{\Tr A(\ep)}\right)=\frac{1}{\Tr A(\ep)} S(A(\ep))+\log \Tr A(\ep)\nonumber\\
&= S(A(\ep))+\ep^2\left[-\Tr(yy^*)S(xx^*)+\Tr(yy^*)\right]+O(\ep^3).
\end{align}
Next we claim
\begin{align*}
& S(A(\ep)) =-\sum_{i=1}^m \lambda_i(\ep)\log \lambda_i(\ep)\\
& =-\sum_{i=1}^m (\mu_i(\ep)+f_{ii}\ep^2)\log (\mu_i(\ep)+f_{ii}\ep^2) +O(\ep^3)=\\
& -\sum_{i=1}^m \mu_i(\ep)\log \mu_i(\ep)-\ep^2(\sum_{i=1}^m f_{ii}\log \lambda_i+\sum_{i=1}^m f_{ii})+O(\ep^3)\\
& =S(A_1(\ep))-\ep^2(\Tr ((yy^*)\log (xx^*))+\Tr(yy^*))+O(\ep^3).
\end{align*}
Combine this expression with the expression above it to deduce \eqref{affparfor1}.  \qed

Note that the expression $\Tr\left[(xx^*-yy^*) \log xx^*\right]$ can be either positive or negative.
In the following we give a very simple reason why we can not ignore this term (i.e. use the affine approximation), which also yields
a necessary condition $xx*$ must satisfy if $x$ is a local minimum.

Assume that we have an affine subspace of the form $A+t B$, where 
$\Tr(A)=1, \Tr (B)=0$.  Here $A=xx^*, B=xy^*+yx^*$ on all $y\in K$ satisfying
the condition $\Tr(B)=0$ and $t$ arbitrary real. Let $\Phi$ be the set of all
$A+tB$ such that $A+tB\ge 0$.  Consider the function $S(C)=-\Tr(C\log C)$
where $C\in \Phi$.  Our assumption that $A$ is a critical point in $\Phi$
for the $S(C)$.  Since $S(C)$ is strictly concave on $\Phi$ it follows that
$A$ is a unique global MAXIMUM on $\Phi$!  So if $A$ was a local minimum
for the $H(x), x\in K, \Tr(xx^*)=1$ it follows that the correction term for
$\ep^2$ that we have must be \emph{strictly} positive .  That is, if $x$ is a local min then
\begin{align*}
& \Tr\left[(xx^*-yy^*) \log xx^*\right]\\
& =S(yy^*)-S(xx^*)+S(yy^*\|xx^*)>0
\end{align*}
for all $y\in x^{\perp}$ (assuming the normalization
$\Tr(yy^*)=1$).

\section{A counter example to real additivity conjecture}
 \label{sec:reals}

During the 2008 American Institute for Mathematics workshop ``Geometry and representation theory" \cite{AIM08}, Leonid Gurvits found a counterexample to the analogue of the additivity conjecture for real (rather than complex) matrices. In this appendix we generalize the counterexample to show that the additivity conjecture fails to hold for real spaces of orthogonal matrices containing the identity: there exist real subspaces $K_1 \subseteq \bR^{m_1 \times n_1}$ and $K_2\subseteq \bR^{m_2 \times n_2}$ such $\rH(K_1 \otimes K_2) < \rH(K_1) + \rH(K_2)$.

 $K\subseteq \bR^{m\times m}$ is called an orthogonal subspace
 if any $0\ne A\in K$ is of the form $aQ$ for some scalar $a$
 and an orthogonal matrix $Q$.  Note that if $K$ is an
 orthogonal subspace then for any orthogonal matrix $Q_0$, the subspace
 $Q_0 \trans K$ is also an orthogonal subspace.  By choosing
 $Q_0 \in  K$ we can always assume that $K$ contains the
 identity matrix $I_m$.
 
The maximal size of an orthogonal subspace is given by the Radon-Hurwitz number, defined as follows.  For $m\in\N$,  let $m=2^b \cdot a$, with $a$ odd, and let $b=4c+d$ where $c$ is a nonnegative integer and $d\in\{0,1,2,3\}$.  Then Radon Hurwitz number of $m$ is $\rho(m):=2^d +8c$.
 
 \begin{theorem}\label{radhur}  Let $K\subseteq \bR^{m\times m}$ be
 an orthogonal subspace.  Then $k :=\dim K\le \rho(m)$, and this
 inequality is sharp for any $m\in\N$.  More precisely, assume
 that $I_m\in K$ and $k \ge 2$.  Then $K$ has a basis
 $I_m,Q_1,\ldots,Q_{k-1}$ where $Q_1,\ldots,Q_{k-1}$ is a set
 of skew symmetric orthogonal anticommuting matrices,
 i.e. $Q_iQ_j=-Q_jQ_i$ for any $1\le i <j\le k-1$.

 Conversely, if $Q_1,\ldots, Q_{k-1}\in \bR^{m\times m}$ are
 $k -1$ skew symmetric orthogonal anticommuting matrices
 then $\spn(I_m,Q_1,\ldots,Q_{k-1})$ is an $k$-dimensional
 orthogonal subspace.

 \end{theorem}

If $Q\in \bR^{m\times m}$ is an
 orthogonal matrix, then all $m$ singular values of $Q$ are equal to
 $1$. Let $Q_i\in \bR^{m_i\times m_i}$ be an orthogonal matrix for
 $i=1,2$.  Then for any real $a_1,a_2$, the singular values of 
 $a_iQ_i$ are $|a_i|$, and the singular values of $(a_1Q_1)\otimes (a_2Q_2)$ are all $|a_1a_2|$.

 Suppose furthermore that $m_1,m_2$ are even and $Q_1, Q_2$ are skew
 symmetric orthogonal matrices.  Then $a_iQ_i$ has $\frac{m_i}{2}$
 eigenvalues equal to $a_i\sqrt{-1}$ and $-a_i\sqrt{-1}$ for $i=1,2$ repectively.
 Furthermore, $(a_1Q)\otimes (a_2Q)$ is a real symmetric
 matrix with $\frac{m_1m_2}{2}$ eigenvalues equal to $a_1a_2$ and
 $\frac{m_1m_2}{2}$ eigenvalues equal to $-a_1 a_2$.

 \begin{theorem}  Let $K\subseteq \bR^{m\times m}$ be an orthogonal
 subspace.  Then $\rH(K)=\log m$.

 Suppose furthermore that $m_1,m_2$ are even and $K_i\subset
 \bR^{m_i\times m_i}$ are orthogonal subspaces of dimension two
 at least for $i=1,2$.  Then
 \begin{align}
 \rH(K_1 \otimes K_2)\le & \log
 \frac{m_1m_2}{2}=\log(m_1m_2)-\log
 2\nonumber\\
 &=\rH(K_1)+\rH(K_2)-\log 2\label{HL12est}
 \end{align}
 In particular, the additivity conjecture does not hold for real subspaces of matrices.

 \end{theorem}
 \proof  Since any matrix $x \in K$ is of the form $aQ$ for some
 orthogonal $Q$ it follows that if $\Tr(xx \trans)=1$ then the singular values of $x$ are all equal to $\frac{1}{m}$.  Hence $\rH(x)=\log m$ and $\rH(K)=\log m$.

 Assume now that $K_1,K_2$ are orthogonal spaces of
 dimension two at least.  Without loss of generality we may
 assume that $I_{m_1},Q_1\in \bR^{m_1\times m_1}$ and $I_{m_2},Q_2\in \bR^{m_2\times m_2}$, where $Q_1,Q_2$ are orthogonal.  Hence
 $I_{m_1 m_2}=I_{m_1}\otimes I_{m_2}$ and $Q_1\otimes Q_2$ are both in $K_1\otimes K_2$.  Recall that $Q_1\otimes Q_2$ is a symmetric matrix which has $\frac{m_1m_2}{2}$ eigenvalues equal to $1$ and $-1$ respectively.  Hence
 $Q_1\otimes Q_2+I_{m_1m_2}$ is a nonnegative definite real symmetric
 matrices which has $\frac{m_1m_2}{2}$
 eigenvalues equal to $2$ and $0$ respectively.
 Let $x=(\frac{2}{m_1m_2})^{\frac{1}{2}}(Q_1\otimes Q_2+I_{m_1m_2})$.  Then $\Tr(xx \trans)=1$
 and $x$ has $\frac{m_1m_2}{2}$ nonzero singular values all equal
 to $(\frac{2}{m_1m_2})^{\frac{1}{2}}$.  Hence
$\rH(K_1\otimes K_2) \le \rH(x)=\log(\frac{m_1m_2}{2})$.
 \qed
\\

\emph{Acknowledgments:---}
We appreciate many valuable discussions with Jon Yard. 
GG and AR research was supported by NSERC, PIMS, and iCORE.

\end{document}